\documentclass[]{article}

\usepackage{authblk}
\usepackage{cite}
\usepackage{amsmath,amssymb,amsthm,amsfonts}
\usepackage{algorithmic}
\usepackage{graphicx}
\usepackage{textcomp}
\usepackage[english]{babel}
\usepackage[utf8]{inputenc}
\usepackage{ragged2e}
\usepackage{comment}
\usepackage{xcolor}
\usepackage[colorlinks = true,linkcolor = blue,urlcolor = blue,citecolor = blue,anchorcolor = blue]{hyperref}
\usepackage{url}
\def\BibTeX{{\rm B\kern-.05em{\sc i\kern-.025em b}\kern-.08em
    T\kern-.1667em\lower.7ex\hbox{E}\kern-.125emX}}
    
\newcommand{\Gate}[1]{\textsc{#1}}

\newcommand{\cnotgate}{\Gate{CNOT}} 
\newcommand{\czgate}{\Gate{CZ}}

\newcommand{\hgate}{\Gate{H}}

\newcommand{\pgate}{\Gate{P}}

\newtheorem{proposition}{Proposition}
\newtheorem{lemma}{Lemma}

\newtheorem{theorem}{Theorem}
\newtheorem{example}{Example}

\newcommand{\eq}[1]{Eq.~(\ref{eq:#1})}
\renewcommand{\sec}[1]{\hyperref[sec:#1]{Section~\ref*{sec:#1}}}
\newcommand{\ssec}[1]{\hyperref[ssec:#1]{Subsection~\ref*{ssec:#1}}}
\newcommand{\fig}[1]{\hyperref[fig:#1]{Figure~\ref*{fig:#1}}}
\newcommand{\tab}[1]{\hyperref[tab:#1]{Table~\ref*{tab:#1}}}
\newcommand{\lem}[1]{\hyperref[lem:#1]{Lemma~\ref*{lem:#1}}}
\newcommand{\prop}[1]{\hyperref[prop:#1]{Proposition~\ref*{prop:#1}}}
\newcommand{\thm}[1]{\hyperref[thm:#1]{Theorem~\ref*{thm:#1}}}

\usepackage{xy}
\xyoption{matrix}
\xyoption{frame}
\xyoption{arrow}
\xyoption{arc}

\usepackage{ifpdf}
\ifpdf
\else
\PackageWarningNoLine{Qcircuit}{Qcircuit is loading in Postscript mode.  The Xy-pic options ps and dvips will be loaded.  If you wish to use other Postscript drivers for Xy-pic, you must modify the code in Qcircuit.tex}
\xyoption{ps}
\xyoption{dvips}
\fi

\entrymodifiers={!C\entrybox}

\newcommand{\ket}[1]{{\left\vert{#1}\right\rangle}}
\newcommand{\qw}[1][-1]{\ar @{-} [0,#1]}
\newcommand{\qwx}[1][-1]{\ar @{-} [#1,0]}

\newcommand{\gate}[1]{*+<.6em>{#1} \POS ="i","i"+UR;"i"+UL **\dir{-};"i"+DL **\dir{-};"i"+DR **\dir{-};"i"+UR **\dir{-},"i" \qw}

\newcommand{\control}{*!<0em,.025em>-=-<.2em>{\bullet}}

\newcommand{\ctrl}[1]{\control \qwx[#1] \qw}

\newcommand{\targ}{*+<.02em,.02em>{\xy ="i","i"-<.39em,0em>;"i"+<.39em,0em> **\dir{-}, "i"-<0em,.39em>;"i"+<0em,.39em> **\dir{-},"i"*\xycircle<.4em>{} \endxy} \qw}
\newcommand{\lstick}[1]{*!R!<.5em,0em>=<0em>{#1}}

\newcommand{\Qcircuit}{\xymatrix @*=<0em>}

\title{Depth optimization of $\czgate$, $\cnotgate$, and Clifford circuits}
\author{Dmitri Maslov and Ben Zindorf}
\affil{\small{IBM Quantum, IBM Thomas J. Watson Research Center, Yorktown Heights, NY 10598}}
\date{\today}

\begin{document}
\maketitle

\begin{abstract}
We seek to develop better upper bound guarantees on the depth of quantum $\czgate$ gate, $\cnotgate$ gate, and Clifford circuits than those reported previously.  We focus on the number of qubits $n\,{\leq}\,$1,345,000 \cite{de2021reducing}, which represents the most practical use case.  Our upper bound on the depth of $\czgate$ circuits is $\lfloor n/2 + 0.4993{\cdot}\log^2(n) + 3.0191{\cdot}\log(n) - 10.9139\rfloor$, improving best known depth by a factor of roughly 2.  We extend the constructions used to prove this upper bound to obtain depth upper bound of $\lfloor n + 1.9496{\cdot}\log^2(n) + 3.5075{\cdot}\log(n) - 23.4269 \rfloor$ for $\cnotgate$ gate circuits, offering an improvement by a factor of roughly $4/3$ over state of the art, and depth upper bound of $\lfloor 2n + 2.9487{\cdot}\log^2(n) + 8.4909{\cdot}\log(n) - 44.4798\rfloor$ for Clifford circuits, offering an improvement by a factor of roughly $5/3$.
\end{abstract}

\section{Introduction} \label{sec:introduction}
Clifford circuits play an important role in quantum computing.  Most prominently, they lie at the core of quantum error correction \cite{nielsen2002quantum}, where they are responsible for both state encoding and state/gate distillation \cite{bravyi2005universal}.  Once error corrected, fault-tolerant computations are often expressed as Clifford+T circuits, directly implying that large chunks of such computations are themselves Clifford circuits.  Clifford circuits play a key role in randomized benchmarking of quantum gates \cite{knill2008randomized, magesan2011scalable}, the study of entanglement \cite{bennett1996mixed}, and shadow tomography \cite{aaronson2020shadow} to name a few more areas of importance.

Clifford circuits can be defined as those quantum transformations computable by the quantum circuits using single-qubit Hadamard gate $\hgate\,{:=}\,\frac{1}{\sqrt{2}}\big(\begin{smallmatrix}1 & 1\\1 & -1 \end{smallmatrix}\big)$, single-qubit Phase gate $\pgate\,{:=}\,\big(\begin{smallmatrix}1 & 0\\0 & i \end{smallmatrix}\big)$, and the entangling $\cnotgate$ gate $\cnotgate\,{:=}\,\left(\begin{smallmatrix}1 & 0 & 0 & 0\\0 & 1 & 0 & 0\\0 & 0 & 0 & 1\\0 & 0 & 1 & 0\end{smallmatrix}\right)$.  In this work, we also utilize the two-qubit $\czgate$ gate which can be defined directly as the transformation $\czgate\,{:=}\,\left(\begin{smallmatrix}1 & 0 & 0 & 0\\0 & 1 & 0 & 0\\0 & 0 & 1 & 0\\0 & 0 & 0 & -1\end{smallmatrix}\right)$ or constructed as a well-known circuit with two Hadamard and one $\cnotgate$ gates. 


Circuit implementations of Clifford transformations have been studied well in the relevant literature.  Optimal Clifford circuits are known for up to 6 qubits \cite{bravyi20206}; however, optimal synthesis of Clifford circuits spanning more than 6 qubits appears to be out of reach.  Asymptotically optimal circuit constructions of arbitrary $n$-qubit Clifford computations are known: a Clifford operation can be implemented with $\Theta\left(n^2/\log(n)\right)$ entangling gates \cite{aaronson2004improved, patel2003efficient} in depth $\Theta\left(n/\log(n)\right)$ \cite{patel2003efficient, jiang2020optimal, de2021reducing}.  No better guarantees, such as asymptotic tightness---meaning asymptotic equality discarding the lower order additive terms, however, are known.  

Due to the 11-stage layered decomposition \cite{aaronson2004improved} over the gate library $\{\hgate,\pgate,\cnotgate\}$, asymptotic analysis of the depth of Clifford circuits relies on the bounds for $\cnotgate$ gate circuits, also known as linear reversible circuits.  The $\cnotgate$ circuit synthesis algorithm offering asymptotically optimal upper bound comes with a high leading constant of $20$---specifically, the depth complexity guarantee \cite{patel2003efficient, jiang2020optimal, de2021reducing} is $\frac{20n}{\log(n)} + O(\sqrt{n}\log(n))$.  Given depth-$2n$ implementation was known since 2007 \cite{kutin2007computation}, it became clear that the asymptotically optimal implementation does not offer an advantage until $n$ becomes very large.  The authors of \cite{de2021reducing} addressed this by introducing an algorithm offering depth upper bound of $4n/3 + 8\lceil\log(n)\rceil$, that outperforms the asymptotically optimal algorithm \cite{patel2003efficient, jiang2020optimal} for $n\,{<}\,1{,}345{,}000$ and outperforms Kutin's et al. algorithm \cite{kutin2007computation} when $n\,{>}\,75$.  One of the results that we report here is a synthesis algorithm that outperforms the combination of all of the above for $70 \,{\leq}\, n \,{\leq}\, 1{,}345{,}000$ while offering the upper bound guarantee of $\lfloor n + 1.9496{\cdot}\log^2(n) + 3.5075{\cdot}\log(n) - 23.4269 \rfloor$---roughly a 25\% reduction over \cite{de2021reducing}, see \lem{cnot}.

In this work, we focus on the $\czgate$, $\cnotgate$, and Clifford circuits spanning no more than $1{,}345{,}000$ qubits.  The number $1{,}345{,}000$ itself originates from \cite{de2021reducing}.  We believe this bound on the number of qubits $n$ covers all useful use cases for $\czgate$, $\cnotgate$, and Clifford circuits.  Indeed, to put this number in perspective, error-correcting codes often span dozens to hundreds of qubits (thousands and tens of thousands are possible albeit regarded to be on the high side \cite{campbell2017roads}), quantum simulations of condensed matter systems need to rely on only slightly more than 54 qubits before they become classically intractable \cite{pednault2019leveraging}, and known likely classically difficult simulations require as few as between 70 and 185 \cite{nam2019low} or between 109 and 111 \cite{reiher2017elucidating} qubits.  To factor a $1000+$ bit integer number using Shor's algorithm---a task widely believed to be intractable classically---only (roughly) $2n$ to $3n$ qubits suffice \cite{beauregard2002circuit, gidney2021factor}.  This qubit count takes additional space needed for high-quality circuit optimization into account.  This points to the high likelihood that the number of qubits a Clifford circuit spans will remain well under $1{,}345{,}000$.

Our goal is to minimize the depth of quantum circuits, which corresponds to time to solution, being perhaps the single most important metric from the consumer's point of view (especially once the fidelity is guaranteed).  Furthermore, in quantum information processing technologies, such as superconducting circuits, where the dominating source of errors is the decoherence, small depth circuits naturally improve the fidelity of the computation compared to large depth circuits.  We measure the depth of circuits by counting the contribution from the two-qubit gates and discarding that from the single-qubit gates.  There are two basic reasons to make this choice.  First, both leading quantum information processing technologies, superconducting circuits and trapped ions, offer single-qubit gates at a much higher clock speed and fidelity compared to the two-qubit gates \cite{IBMQ, debnath2016demonstration}.  Due to available control and as motivated by Euler's angle decomposition, the number of single-qubit pulses applied between the entangling gates is never more than a small constant (e.g., 3).  Thus, the depth by the two-qubit gates describes the depth of the real-life physical implementation rather closely.  Second, the entangling gates we rely on, $\cnotgate$ and $\czgate$, are single-qubit equivalent to each other, each can be obtained with the minimal number of one entangling pulse in both superconducting circuits and trapped ions technologies, and neither of the two directly corresponds to the physical qubit-to-qubit interaction (such as ZX in superconducting circuits and XX in trapped ions \cite{IBMQ, debnath2016demonstration}).  Thus, both $\cnotgate$ and $\czgate$ gates are available simultaneously, and their implementation costs are roughly equal---independently of the underlying technology used to implement the desired circuits.

Our work first focuses on the $\czgate$ circuits.  $\czgate$ circuits are employed in the short layered decomposition of Clifford circuits \cite{bravyi2021hadamard}, thus allowing to upper bound the depth of Clifford circuits more efficiently than would otherwise be possible with the reduction of -CZ- layers to -CNOT- and -P- layers.  A $\czgate$ circuit can be implemented over $\czgate$ gates in depth $n{-}1$ for even $n$ and depth $n$ for odd $n$.  This can be established directly, or by employing Vizing's theorem \cite{vizing1964estimate}.  One may also show that the depth cannot be reduced further unless other gates are allowed.  In our work, we employ $\cnotgate$ gates and show how this helps to reduce the depth of $\czgate$ circuits roughly by a factor of $2$ (\thm{1}).  We utilize depth-efficient implementations of $\czgate$ circuits to construct depth-optimized $\cnotgate$ and Clifford circuits.

\section{Circuit depth guarantees}

\subsection{CZ circuits} \label{ssec:cz}

We first focus on the depth-efficient no ancilla implementation of the elements of the finite group generated by $\czgate$ gates over $n$ qubits. Recall the following well-known properties of $\czgate$ gates: $\czgate(i,j) \,{=}\, \czgate(j,i)$, $\czgate(i,j)^2$ equals the identity, and all $\czgate$ gates commute.  These properties directly imply that any $\czgate$ circuit can be represented by a zero-diagonal upper triangular binary matrix $M \in \mathbb{F}_2^{n\times n}$, where $m_{i,j}\,{=}\,1$ for $i\,{<}\,j$ iff the gate $\czgate(i,j)$ is applied (an odd number of times).  The task of implementing a transformation described by the matrix $M$ can therefore be solved by applying a set of gates that zero out all of the entries of matrix $M$. 

We first focus on developing a small-depth circuit implementing a $\czgate$ transformation $M1$ that can be described by a ``rectangular'' $k\,{\times}\,m$ region (over non-overlapping sets of $k$ and $m$ qubits) with ones in the matrix $M$; the rest of the matrix $M$ elements are zeroes.  A straightforward implementation of such transformation can be accomplished in depth $\max\{k,m\}$ by a circuit with $km$ $\czgate$ gates.  Our construction described below thus offers an exponential advantage over the na\"{i}ve implementation.  Formally, 

\begin{lemma}\label{lem:rectangles}
Let $A:=\{a_1,a_2,...,a_k\}$ and $B:=\{b_1,b_2,...,b_m\}$ be non-overlapping sets of qubits.  A $\czgate$ transformation $M1$ described as the set of gates $\czgate(a,b)$ for all $a \,{\in}\, A$ and $b \,{\in}\, B$ can be implemented in depth $2{\cdot}\max\{\lceil \log(k) \rceil, \lceil \log(m) \rceil\}$. 
\end{lemma}

\begin{proof}
First, recall that the action of $\czgate(x,y)$ is accomplished by the mapping $\ket{x,y} \mapsto (-1)^{xy}\ket{x,y}$, i.e., it can be described as the addition of phase $(-1)^{xy}$ to $\ket{x,y}$.  Thus, the phase transformation performed by $M1$ is described as  

\begin{eqnarray*}
\prod_{a \in A, b\in B}(-1)^{a\cdot b} = (-1)^{\left( \displaystyle\bigoplus_{i=1..k,j=1..m} a_ib_j \right)} 
= (-1)^{(a_1\oplus a_2 \oplus ... \oplus a_k)(b_1\oplus b_2 \oplus ... \oplus b_m)}.
\end{eqnarray*}

\noindent The latter term can be implemented by a single $\czgate$ gate acting on qubits carrying the values $a_1\oplus a_2 \oplus ... \oplus a_k$ and $b_1\oplus b_2 \oplus ... \oplus b_m$.  Those linear combinations can be implemented in logarithmic depth (to both compute them and uncompute after applying $\czgate$) by a $\cnotgate$ gate circuit, leading to the overall depth of $2{\cdot}\max\{\lceil \log(k) \rceil, \lceil \log(m) \rceil\}+1$. 

We next explain how to reduce the depth by $1$, leading to the advertised complexity. To accomplish the reduction, we focus on the three central layers of the constructed circuit.  Observe that the middle gate is always a single $\czgate$, and logarithmic-depth EXOR (exclusive OR, also known as modulo two addition) calculation of qubits in the sets $A$ and $B$ ends with a single $\cnotgate$ gate.  Because of the varying depths of the $\cnotgate$ parts for sets $A$ and $B$, the three middle stages come in the following three flavors,

\begin{center}
	\includegraphics[width=0.48\textwidth]{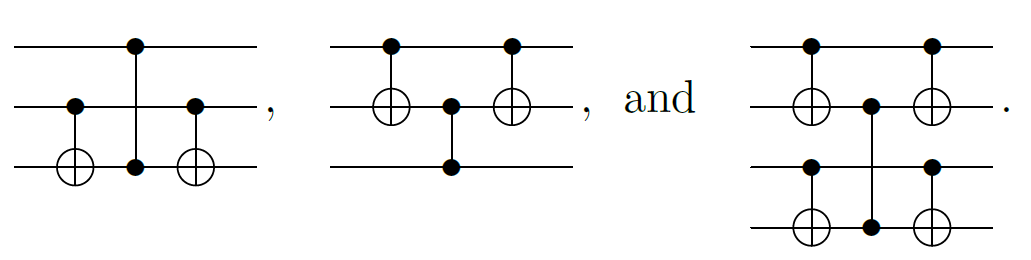}
\end{center}

Each can correspondingly be rewritten in depth two, as follows:

\begin{center}
	\includegraphics[width=0.4\textwidth]{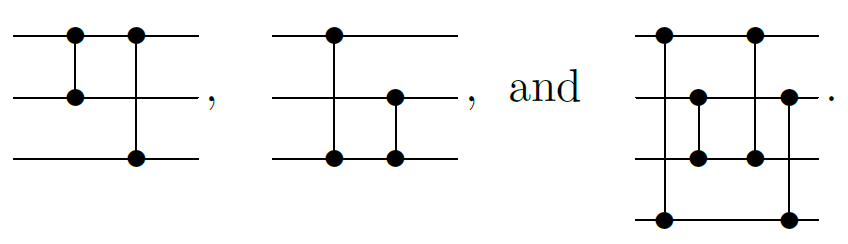}\par
\end{center}

We illustrated the resulting circuit in \fig{45rectangle} for $k{=}4$ and $m{=}5$. 
\end{proof}

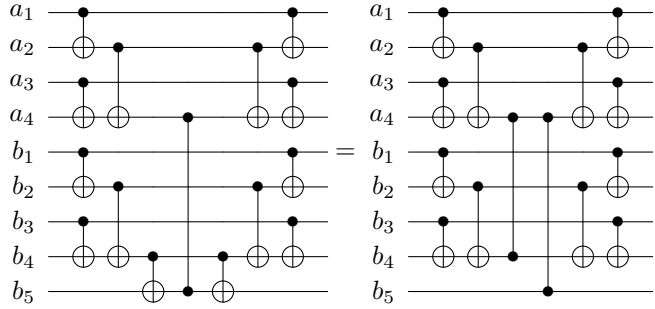
\begin{figure}[t]
\[ 
\Qcircuit @C=0.5em @R=0.5em @!{
\lstick{a_1} & \ctrl{1}  & \qw       & \qw       & \qw       & \qw       & \qw       & \ctrl{1}  & \qw \\
\lstick{a_2} & \targ     & \ctrl{2}  & \qw       & \qw       & \qw       & \ctrl{2}  & \targ     & \qw \\
\lstick{a_3} & \ctrl{1}  & \qw       & \qw       & \qw       & \qw       & \qw       & \ctrl{1}  & \qw \\
\lstick{a_4} & \targ     & \targ     & \qw       & \ctrl{5}  & \qw       & \targ     & \targ     & \qw \\
\lstick{b_1} & \ctrl{1}  & \qw       & \qw       & \qw       & \qw       & \qw       & \ctrl{1}  & \qw \\
\lstick{b_2} & \targ     & \ctrl{2}  & \qw       & \qw       & \qw       & \ctrl{2}  & \targ     & \qw \\
\lstick{b_3} & \ctrl{1}  & \qw       & \qw       & \qw       & \qw       & \qw       & \ctrl{1}  & \qw \\
\lstick{b_4} & \targ     & \targ     & \ctrl{1}  & \qw       & \ctrl{1}  & \targ     & \targ     & \qw \\
\lstick{b_5} & \qw       & \qw       & \targ     & \ctrl{0}  & \targ     & \qw       & \qw       & \qw
}
\hspace{1mm}\raisebox{-18.5mm}{=}\hspace{7mm}
\Qcircuit @C=0.5em @R=0.5em @!{
\lstick{a_1} & \ctrl{1}  & \qw       & \qw       & \qw       & \qw       & \ctrl{1}  & \qw \\
\lstick{a_2} & \targ     & \ctrl{2}  & \qw       & \qw       & \ctrl{2}  & \targ     & \qw \\
\lstick{a_3} & \ctrl{1}  & \qw       & \qw       & \qw       & \qw       & \ctrl{1}  & \qw \\
\lstick{a_4} & \targ     & \targ     & \ctrl{4}  & \ctrl{5}  & \targ     & \targ     & \qw \\
\lstick{b_1} & \ctrl{1}  & \qw       & \qw       & \qw       & \qw       & \ctrl{1}  & \qw \\
\lstick{b_2} & \targ     & \ctrl{2}  & \qw       & \qw       & \ctrl{2}  & \targ     & \qw \\
\lstick{b_3} & \ctrl{1}  & \qw       & \qw       & \qw       & \qw       & \ctrl{1}  & \qw \\
\lstick{b_4} & \targ     & \targ     & \ctrl{0}  & \qw       & \targ     & \targ     & \qw \\
\lstick{b_5} & \qw       & \qw       & \qw       & \ctrl{0}  & \qw       & \qw       & \qw
}
\]
\caption{Depth-6 implementation of a ``rectangular'' set of gates $\{\czgate(a,b)\}$ for all $a \in A\,{=}\,\{a_1,a_2,a_3,a_4\}$ and $b \in B\,{=}\,\{b_1,b_2,b_3,b_4,b_5\}$: left side shows basic circuit construction, and right side includes the reduction of the depth by $1$.}
\label{fig:45rectangle}
\end{figure}

We next focus on a more complex version of the rectangular region $M1$ defined as the transformation $M01$ computed by a {\em subset} of $\czgate(a,b)$ (rather than all for the case of $M1$), where $a \in A$, $b \in B$, and $A \cap B = \emptyset$.  We show that $M01$ can be implemented in depth $\max\{\lfloor k/2\rfloor,\lfloor m/2 \rfloor\} + 2{\cdot}\max\{\lceil\log(k)\rceil,\lceil\log(m)\rceil\}$.

\begin{lemma}\label{lem:M01}
The transformation $M01$ over non-overlapping sets $A$ and $B$ with $k$ and $m$ qubits each can be implemented as a depth $\max\{\lfloor k/2\rfloor,\lfloor m/2 \rfloor \} + 2{\cdot}\max\{\lceil\log(k)\rceil,\lceil\log(m)\rceil\}$ circuit.
\end{lemma}
\begin{proof}
The transformation $M01$ can be written as a Boolean matrix $\{m01_{i,j}\}|_{i=1..k,j=1..m}$, where $m01_{i,j}$ denotes the presence (1) or the absence (0) of the gate $CZ(a_i,b_j)$.  By a slight abuse of language $M01$ can be interpreted as a rectangle $A{\times}B$ with zeroes and ones.  To implement $M01$ as an efficient circuit, we apply a logarithmic depth circuit that reduces $M01$ to a transformation $M01^\prime$ such that the weight (the number of ones) of rows in it is no more than $\lfloor m/2 \rfloor$ and the weight of columns is no more than $\lfloor k/2 \rfloor$. Since $M01^\prime$ can be interpreted as an adjacency matrix of a bipartite graph, the edge coloring problem can be solved using precisely $\max\{\lfloor k/2\rfloor,\lfloor m/2 \rfloor \}$ colors \cite{vizing1964estimate, cole1982edge}.  Edges of the same color correspond to individual $\czgate$ gates implementable in depth $1$, and thus the number of colors describes the $\czgate$ gate circuit depth.

We next show how to reduce $M01$ to $M01^\prime$ by a logarithmic depth circuit.  To this end, we first show how to select a set of rows and columns of the matrix $M01$ that, when inverted, simultaneously reduce the row and column weights to no more than a half, and next express this row and column inversion transformation as a logarithmic depth circuit.

To select rows and columns, start with the empty set $S$.  Cycle sequentially through all rows and columns in an infinite loop.  If inverting an entire given row/column reduces the number of ones in it, perform the inversion and add this row/column to the set $S$, or if it is already there remove it.  Each row/column addition/removal operation reduces the number of ones in the matrix $M01$ by at least one, thus this algorithm will run out of options to invert a row/column and thus can be terminated after no more than $km(k{+}m)$ steps.  When it terminates, $M01$ has been transformed to $M01^\prime$ with row and column weights of no more than a half.

Denote the sets of rows and columns identified in the previous paragraph as $A^\prime$ and $B^\prime$, correspondingly.  To implement this set of row and column flips, observe that rather than implementing the rectangles $A \times B^\prime$ (implements all columns) and $A^\prime \times B$ (implements all rows) sequentially, one could instead implement the rectangles $A{\setminus}A^\prime \times B^\prime$ and $A^\prime \times B{\setminus}B^\prime$ in parallel, since the qubit sets $A^\prime$, $A{\setminus}A^\prime$, $B^\prime$, and $B{\setminus}B^\prime$ do not overlap.  According to \lem{rectangles}, this can be done in depth

\begin{eqnarray*}
2{\cdot}\max\{\lceil \log(|A^\prime|) \rceil, \lceil\log(|A{\setminus}A^\prime|)\rceil, \lceil\log(|B^\prime|)\rceil, \lceil\log(|B{\setminus}B^\prime|)\rceil\}
\leq 2{\cdot}\max\{\lceil\log(k)\rceil,\lceil\log(m)\rceil\}.
\end{eqnarray*}

Adding the cost of the transformation $M01 \mapsto M01^\prime$ with the cost of the implementation of $M01^\prime$ reveals the desired depth figure.
\end{proof}

We now have enough instrument to prove the main result of this section.

\begin{theorem}\label{thm:1}
For $n \,{\in}\, [39..1{,}345{,}000]$ an $n$-qubit $\czgate$ transformation $M$ can be implemented by a depth $\lfloor n/2 + 0.4993{\cdot}\log^2(n) + 3.0191{\cdot}\log(n) - 10.9139\rfloor$ circuit.
\end{theorem}
\begin{proof}
Let $d(n)$ denote the depth of $\czgate$ circuits over $n$ qubits.  We start the proof by recalling that an $n$-qubit $\czgate$ circuit can be implemented in depth $n-\left\lceil\left\lceil\frac{n{+}1}{2}\right\rceil{-}\frac{n{+}1}{2}\right\rceil$ by the reduction to graph coloring problem \cite{vizing1964estimate}, and thus a simple upper bound holds,
\begin{equation}\label{eq:cz1}
d(n) \leq n-\left\lceil\left\lceil\frac{n{+}1}{2}\right\rceil{-}\frac{n{+}1}{2}\right\rceil.
\end{equation}
For odd $n$, the maximal number of colors, as given by the Vizing's theorem \cite{vizing1964estimate}, is needed.  For even $n$, a widely known simple geometric construction shows that $n{-}1$ colors suffice: take $n{-}1$ points on the plane as vertices of a regular polygon, with the last $n$-th point at its center.  Each of the $n{-}1$ colors applies to the segment joining the point at the center with a selected vertex of the polygon, and all segments joining polygon vertices perpendicular to it.  One may convince themselves that all possible segments joining any two of the $n$ points considered are properly colored, and thus $n{-}1$ colors suffice.  Furthermore, if one is limited to using the $\czgate$ gates to implement $\czgate$ circuits, the bound in \eq{cz1} is tight.  This follows from the counting argument, noting that the largest $\czgate$ circuit contains $\frac{(n-1)n}{2}$ $\czgate$ gates.  Thus, to implement $\czgate$ circuits in shorter depth, one must thus rely on other gates, which is what we do.

We next introduce a recursive construction that is responsible for reducing the above depth figure to almost $n/2$ and analyze it carefully using two methods.  In our recursion, at each step the set of qubits is broken into two non-overlapping sets, $A$ with first $\lceil n/2 \rceil$ qubits and $B$ with last $\lfloor n/2 \rfloor$ qubits.  Operation $M$ can be expressed in three parts: $M$ restricted to the set $A$, $M$ restricted to the set $B$, and $M01$ over the rectangle $A{\times B}$.  Since the first two can be implemented in parallel, the overall depth can be upper bounded as 
\[
d(n) \leq d(\lceil n/2 \rceil) + d(M01),   
\]
where $d(M01)$ is the depth of the implementation of $M01$.  In other words, by \lem{M01},
\[
d(n) \leq d(\lceil n/2 \rceil) + \left\lfloor \lceil n/2 \rceil /2 \right\rfloor + 2{\cdot}\lceil\log(n/2)\rceil.
\]
Combining the above with \eq{cz1} allows to obtain the following recursion, 
\begin{equation}\label{eq:cz2}
\begin{cases}
d(n) \leq  \text{MIN} \{ n-\left\lceil\left\lceil\frac{n{+}1}{2}\right\rceil{-}\frac{n{+}1}{2}\right\rceil, 
\\ \quad\quad\quad\quad\quad d(\lceil n/2 \rceil) + \left\lfloor \lceil n/2 \rceil /2 \right\rfloor + 2{\cdot}\lceil\log(n/2)\rceil \} \\
d(2) =1, \; d(3) =3.    
\end{cases}
\end{equation}

The solution to \eq{cz2} can be upper bounded by the expression $\lfloor n/2 + 0.9937{\cdot}\log^2(n) + 1.1882{\cdot}\log(n) - 14.6772 \rfloor$ (for $n \,{\in}\, [43..1{,}345{,}000]$).  However, the constant in front of $\log^2(n)$ can be improved through a more careful analysis of the recursive decomposition based on \lem{M01}.  We accomplish this by considering two steps of the recursive decomposition at once.  

Each recursive step implements the transformation $T{:}\,M01 \mapsto M01^\prime$, that we further refer to as T-transformation, and the leftover operation $M01^\prime$.  The circuit obtained by two steps of the decomposition can be thought of as a combination of the implementations of two layers of T-transformations (one of which applies two T-transformations to non-overlapping sets) performing the mappings over recursively defined $M01/M01^\prime$ and two layers of the implementations of $M01^\prime$ (one of which applies to two non-overlapping qubit sets) via bipartite graph coloring.  Recall that all four stages implement certain $\czgate$ gate transformations and thus they all commute.  We will employ the commutation property to prove a better bound on the depth of the $\czgate$ circuit.  Specifically, we group the implementations of $M01^\prime$ and all T-transformations into two subcircuits and analyze their depths separately.

The depth of the implementations of two recursively defined layers of $M01^\prime$ is described by the formula 
\[
\lfloor\lceil n/2\rceil/2\rfloor + \lfloor\lceil\lceil n/2\rceil/2\rceil/2\rfloor.
\]

To analyze the depth of two T-transformation layers, recall what transformations they perform.  At the first step of the recursion sets $A$ and $B$ are defined as the first and second halves of the set of variables.  Subsets $A^\prime \,{\subset}\, A$ and $B^\prime \,{\subset}\, B$ are constructed, and the transformation T implements the sets of two all-one rectangles, 
\[
A^\prime \times B{\setminus}B^\prime \text{ and }A{\setminus}A^\prime \times B^\prime,
\]
in parallel, by computing EXORs of variables in the sets $A^\prime$, $A{\setminus}A^\prime$, $B^\prime$, and $B{\setminus}B^\prime$ in logarithmic depth.  At the second step of the recursion sets $AA$, $AB$, $BA$, and $BB$ such that $AA \sqcup AB = A$ and $BA \sqcup BB = B$ are defined with a quarter of the number of qubits in each.  Their subsets $AA^\prime$, $AB^\prime$, $BA^\prime$, and $BB^\prime$ are identified such that the all-one rectangles 

\begin{eqnarray*}
AA^\prime \times AB{\setminus}AB^\prime, \; AA{\setminus}AA^\prime \times AB^\prime, \;
BA^\prime \times BB{\setminus}BB^\prime,
\text{ and } BA{\setminus}BA^\prime \times BB^\prime
\end{eqnarray*}
can be implemented in parallel, since no two sets intersect.  

To implement these two sets of T-transformations, we define $16$ indexed sets $S_{i,j,k}$, where $i$ and $j$ offer $2$ options each, and $k$ offers $4$ options, as follows: $i$ chooses the set $S$ between $A$ and $B$, $j$ chooses between $S^\prime$ and $S{\setminus}S^\prime$, and $k$ chooses between $SA^\prime$, $SA{\setminus}SA^\prime$, $SB^\prime$, and $SB{\setminus}SB^\prime$.  The set $S_{i,j,k}$ is defined as the intersection of the three sets defined by the choice of $i$, $j$, and $k$.  For example, if $i$ chose $A$, $j$ chose $S{\setminus}S^\prime$, and $k$ chose $SB^\prime$, $S_{1,2,3} = A \cap (A{\setminus}A^\prime) \cap AB^\prime = (A{\setminus}A^\prime) \cap AB^\prime$ (here, the enumeration of lists for $i,j,k$ starts with 1).

By definition, no two sets $S_{i,j,k}$ overlap, and each contains no more than $\lceil\lceil n/2\rceil/2\rceil$ qubits.  Thus, EXORs of variables in each can be implemented by a $\cnotgate$ gate circuit in depth at most $\lceil \log \lceil\lceil n/2\rceil/2\rceil \rceil$. The $\czgate$ gate transformations to be applied to these sets can be described as rectangles 
\begin{eqnarray*}
(S_{1,1,1}\oplus S_{1,1,2}\oplus S_{1,1,3}\oplus S_{1,1,4}) 
\times (S_{2,2,1}\oplus S_{2,2,2}\oplus S_{2,2,3}\oplus S_{2,2,4}) \text{ and} \\
(S_{1,2,1}\oplus S_{1,2,2}\oplus S_{1,2,3}\oplus S_{1,2,4})
\times (S_{2,1,1}\oplus S_{2,1,2}\oplus S_{2,1,3}\oplus S_{2,1,4}),
\end{eqnarray*}
applied in parallel, followed by rectangles 
\begin{eqnarray*}
(S_{1,1,1}\oplus S_{1,2,1}) \times (S_{1,1,4}\oplus S_{1,2,4}), \;\;
(S_{1,1,2}\oplus S_{1,2,2}) \times (S_{1,1,3}\oplus S_{1,2,3}), \\
(S_{2,1,1}\oplus S_{2,2,1}) \times (S_{2,1,4}\oplus S_{2,2,4}), \text{ and }
(S_{2,1,2}\oplus S_{2,2,2}) \times (S_{2,1,3}\oplus S_{2,2,3}), \\
\end{eqnarray*}
applied in parallel.  The rectangles are introduced in the same order as they are discussed in the previous paragraph.  Since these are $4{\times}4$ and $2{\times}2$ rectangles, they take total depth $4+2=6$ to implement as a $\czgate$ circuit.  Thus, the total depth to implement T-transformations is $2 {\cdot}\lceil \log \lceil\lceil n/2\rceil/2\rceil \rceil + 6$, and the combined depth of two stages of the recursive decomposition is 
\[
\lfloor\lceil n/2\rceil/2\rfloor + \lfloor\lceil\lceil n/2\rceil/2\rceil/2\rfloor + 2{\cdot} \lceil \log \lceil\lceil n/2\rceil/2\rceil \rceil + 6.
\]

Based on the above analysis, the final form the recursion takes, further improving \eq{cz2}, is 
\begin{equation}\label{eq:cz3}
\begin{cases}
\begin{split}
    d(n) =  \text{MIN} & \{ n-\lceil\lceil(n{+}1)/2\rceil{-}(n{+}1)/2\rceil, \\
        & d(\lceil n/2 \rceil) + \left\lfloor \lceil n/2 \rceil /2 \right\rfloor + 2{\cdot}\lceil\log(n/2)\rceil, \\
        & d(\lceil\lceil n/2\rceil/2\rceil) + \lfloor\lceil n/2\rceil/2\rfloor + \lfloor\lceil\lceil n/2\rceil/2\rceil/2\rfloor + 2{\cdot}\lceil \log \lceil\lceil n/2\rceil/2\rceil \rceil + 6 \}
\end{split}\\
d(1)=0, \;d(2)=1, \; d(3)=3.    
\end{cases}
\end{equation}

We numerically upper bounded the solution to \eq{cz3} by the expression $\lfloor n/2 + 0.4993{\cdot}\log^2(n) + 3.0191{\cdot}\log(n) - 10.9139\rfloor$ for the range of values $n$ of interest.

\end{proof}

We illustrate the comparison of the best previously known bound on the depth of $\czgate$ circuits to the exact solution of the recursion \eq{cz3} and the upper bound given in \thm{1} in \fig{cz}. Note that the exact solution of the recursion \eq{cz3} gives slightly lower numbers than those made available by the upper bound.  This is because the discrete operations ceiling and floor are not easy to model by continuous functions used in the formulation of the upper bound.  At the full scale, see \fig{cz}(b), the difference between exact solution and the upper bound given is visually undetectable, and our result can be seen to improve the best known previously roughly by a factor of two (therefore, agreeing with the asymptotics). 

\begin{figure*}[t]
{\centering
\begin{tabular}{cc}
\includegraphics[width=0.48\textwidth]{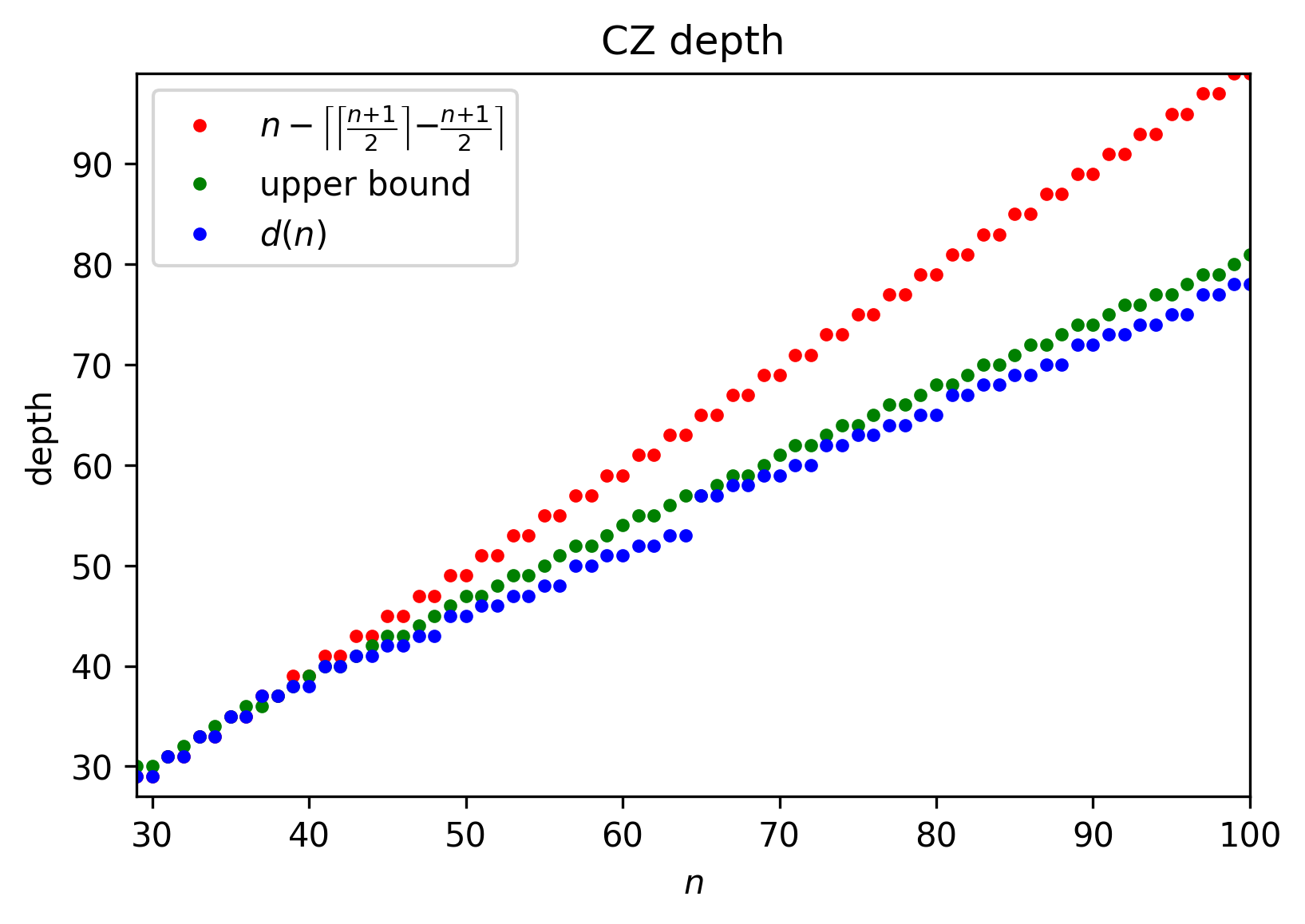} &
\includegraphics[width=0.48\textwidth]{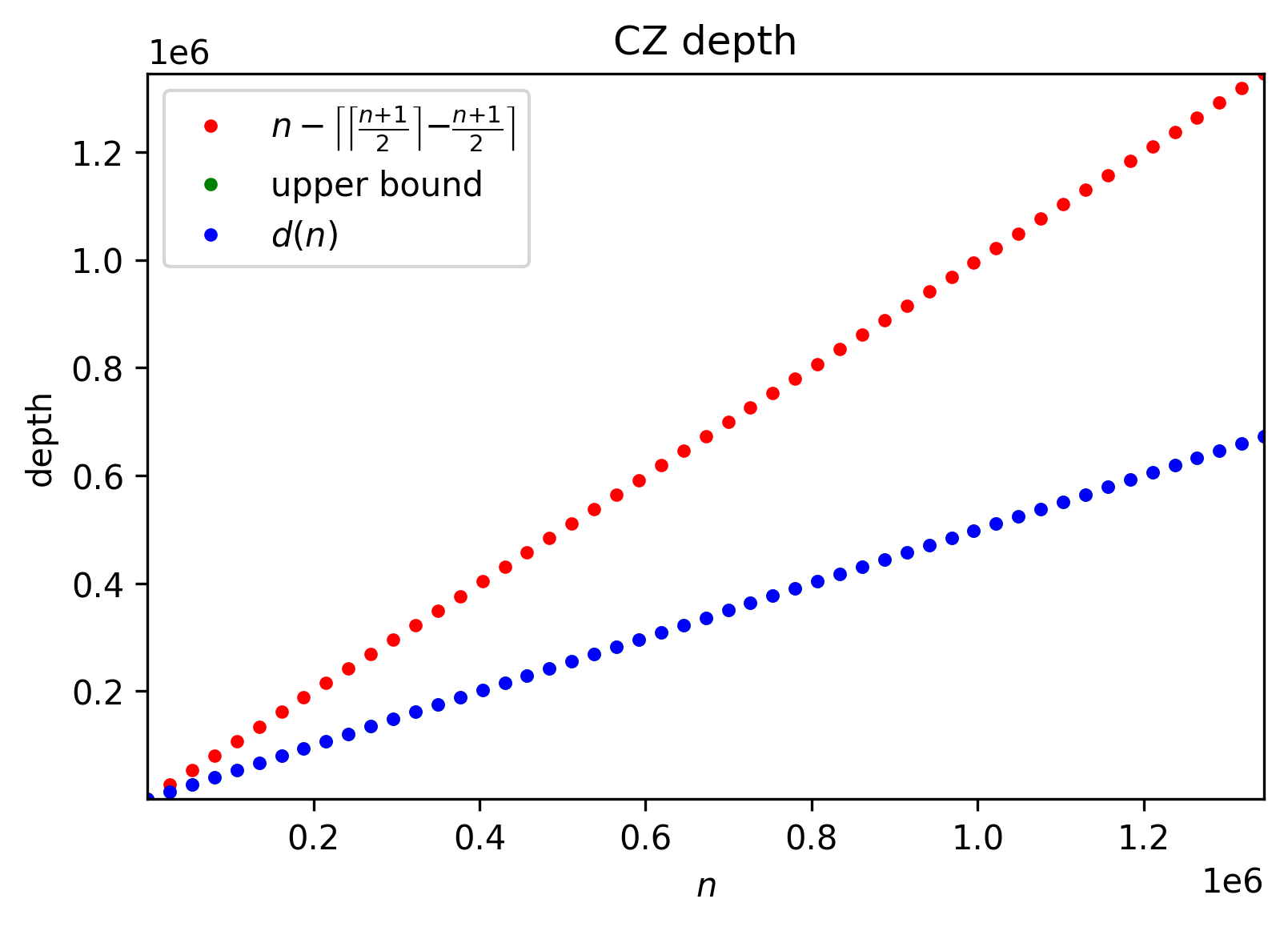} \\
(a) & (b) 
\end{tabular}
}
\caption{Comparison of the best previously known bound  on the $\czgate$ circuit depth (red dots) to the upper bound proved in \thm{1} (green dots) to the solution of the recursion \eq{cz3} (blue dots). (a) focuses on a small number of qubits $n\,{\leq}\,100$ and (b) illustrates the comparison for the full range of values $n$ considered.}
\label{fig:cz}
\end{figure*}


\subsection{CNOT circuits} \label{ssec:cnot}

\begin{figure*}[t]
{\centering
\begin{tabular}{cc}
\includegraphics[width=0.48\textwidth]{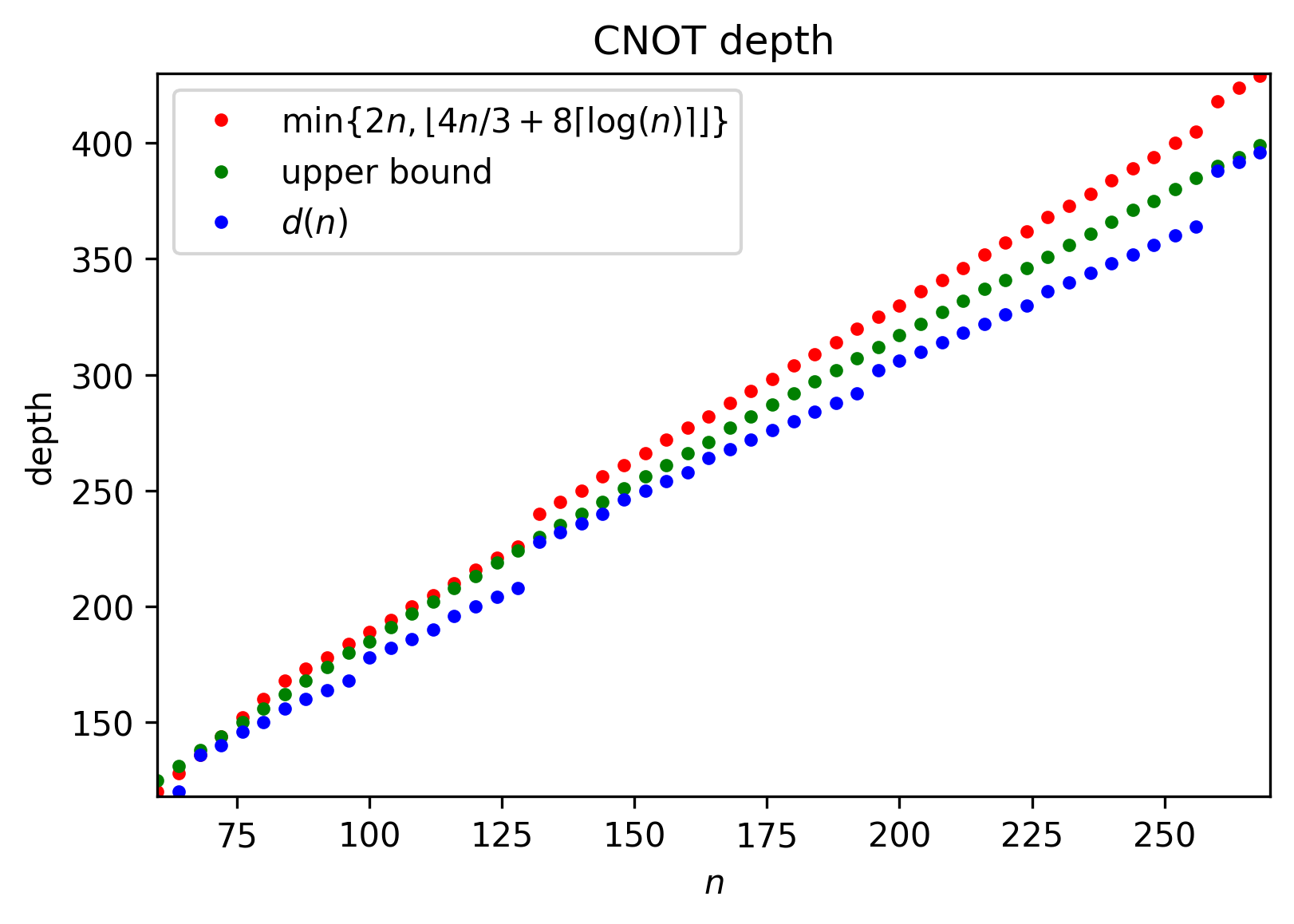} &
\includegraphics[width=0.48\textwidth]{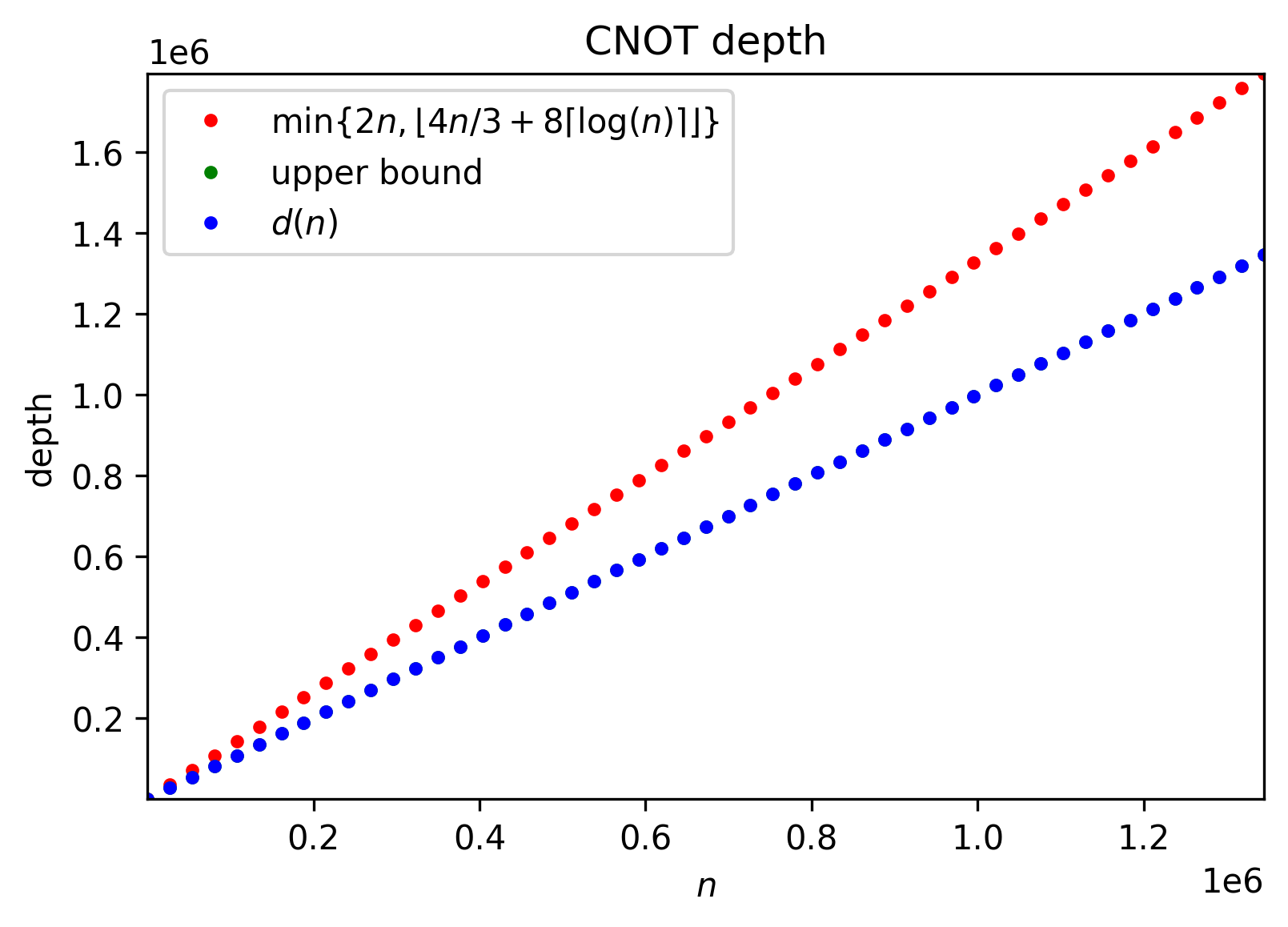} \\
(a) & (b) 
\end{tabular}
}
\caption{Comparison of the best previously known bound on the $\cnotgate$ circuit depth \cite{de2021reducing, kutin2007computation} (red dots) to the upper bound proved in \lem{cnot} (green dots) to the solution of the recursion \eq{cnotrecursion} (blue dots). (a) focuses on a small number of qubits $n\,{<}\,270$ and (b) illustrates the comparison for the full range of values $n$ considered.}
\label{fig:cnot}
\end{figure*}

Here we extend the construction of depth-efficient $\czgate$ circuits to obtain depth-efficient implementations of linear reversible circuits.  A linear reversible function can be implemented exactly or up to the SWAPping of output qubits, also known as qubit reordering.  An implementation up to qubit reordering may be preferred since the proper qubit SWAPping may be obtained classically, allowing to outsource this task to a classical computer and thus minimize the expensive quantum resources used.  The following Lemma reports an optimized depth figure for linear reversible functions and highlights that a depth reduction by $6$ is possible to achieve if it suffices to implement the desired linear function up to qubit reordering.

\begin{lemma}\label{lem:cnot}
For $n\,{\in}\,[70..1{,}345{,}000]$ an $n$-qubit linear reversible transformation $R$ can be implemented in depth no more than $\lfloor n + 1.9496{\cdot}\log^2(n) + 3.5075{\cdot}\log(n) - 29.4269 \rfloor$ up to qubit reordering and depth $\lfloor n + 1.9496{\cdot}\log^2(n) + 3.5075{\cdot}\log(n) - 23.4269 \rfloor$ exactly as a circuit over $\{\cnotgate, \czgate, \hgate\}$ gates.
\end{lemma}
\begin{proof}
We start with the LU decomposition $R\,{=}\,LU$, where $L$ is lower-triangular and $U$ is upper-triangular invertible Boolean matrices.  Recall that the LU decomposition exists subject to proper row and/or column ordering.  Such row/column reordering can be implemented as a SWAPping circuit with the SWAP depth of no more than $2$, translating to the two-qubit gate depth (by those gates considered in this work as contributing to depth) of $6$.  Thus, the difference between the depths of implementations up to qubit reordering and the exact one is a constant equal to $6$.  In the following, we show that each $L$ and $U$ stage can be implemented in depth $\lfloor n/2 + 0.9748{\cdot}\log^2(n) + 1.7538{\cdot}\log(n) - 14.7134\rfloor$, and thus the total depth of $\cnotgate$ circuits is upper bounded by the expression $\lfloor n + 1.9496{\cdot}\log^2(n) + 3.5075{\cdot}\log(n) - 23.4269 \rfloor$. 

Without loss of generality, focus on $U$.  Divide the set of qubits into two, set $A$ with the first $\lceil n/2 \rceil$ qubits and set $B$ with the last $\lfloor n/2 \rfloor$ qubits (this assumes that the qubits are already ordered so as to accept the LU decomposition).  The operation $R$ can be written as the block matrix product 

\begin{equation}\label{eq:3matrices}
R =  \begin{bmatrix} R_A & W01 \\ 0 & R_B \end{bmatrix} =
\begin{bmatrix} R_A & 0 \\ 0 & I \end{bmatrix}
\begin{bmatrix} I & 0 \\ 0 & R_B \end{bmatrix}
\begin{bmatrix} I & R_A^{-1}W01 \\ 0 & I \end{bmatrix},
\end{equation}
where $R_A$ is the $\lceil{n/2}\rceil {\times} \lceil n/2 \rceil$ upper triangular matrix obtained by restricting $R$ to the set of qubits $A$, $R_B$ is defined similarly, $W01$ is the $\lceil n/2 \rceil {\times} \lfloor n/2 \rfloor$ top right block of $R$, and $I$ and $0$ are the identity and zero matrices of proper dimensions.

Assuming $d(n)$ denotes the depth of the implementation of an n-qubit upper triangular matrix, first two terms in the decomposition \eq{3matrices} can be implemented in parallel, i.e. in depth $d\left(\lceil n/2 \rceil\right)$.  The third term can be implemented as the $\cnotgate$ gate circuit where individual gates have targets in the qubit set $A$ and controls in the set $B$.  This transformation can thus be written as the circuit $\hgate_A M01 \hgate_A$, where $\hgate_A$ applies Hadamard gates to all qubits in the set $A$, and $M01$ is an $A{\times}B$ $\czgate$ rectangle.  By inducing the bipartite graph coloring argument, we conclude that the rectangle $M01$ can be implemented in depth at most $\lceil n/2 \rceil$.  This results in the recursion 
\begin{equation}
\begin{cases}
d(n) = d(\lceil n/2 \rceil) + \lceil n/2 \rceil \\
d(2)=1, \;d(3)=2.
\end{cases}    
\end{equation}
The solution, $d^*(n)$, is almost equal to $n$.  For the range of values of interest, we can upper bound it as $d^*(n) \leq n+\lfloor\log(n{-}1)\rfloor-2$.

On the other hand, \lem{M01} can be used to implement $M01$ in depth  $\lfloor \lceil n/2 \rceil/2 \rfloor  + 2{\cdot}\lceil\log(\lceil n/2 \rceil)\rceil$.  Thus, the recursion describing the overall implementation depth can be written as  
\begin{equation}\label{eq:cnotrecursion}
\begin{cases}
d(n) = \text{MIN} \{ d(\lceil n/2 \rceil) + \lceil n/2 \rceil, \\
\quad\quad\quad\quad\quad
d(\lceil n/2 \rceil) + \lfloor \lceil n/2 \rceil/2 \rfloor  + 2{\cdot}\lceil\log(\lceil n/2 \rceil)\rceil \} \\
d(2)=1, \; d(3)=2.
\end{cases}
\end{equation}
We calculated that the solution of recursion \eq{cnotrecursion} can be upper bounded by the expression 
\[
\lfloor n/2+0.9748{\cdot}\log^2(n) + 1.7538{\cdot}\log(n) - 14.7134\rfloor
\]
for the range of values $n\,{\in}\,[70..1{,}345{,}000]$.

We employ the solution of the recursion \eq{cnotrecursion} within the LU decomposition to obtain the desired upper bound, $\lfloor n + 1.9496{\cdot}\log^2(n) + 3.5075{\cdot}\log(n) - 23.4269 \rfloor$.  We start the range with $n\,{=}\,70$, because it marks the smallest $n$ for which our solution based on the recursion \eq{cnotrecursion} beats the best known upper bound of $\min\{2n,\lfloor 4n/3+8\lceil\log(n)\rceil\rfloor \}$ \cite{de2021reducing, kutin2007computation}.
\end{proof}

Note that the circuit constructed in \lem{cnot} relies on the gates from the library $\{\cnotgate, \czgate, \hgate\}$. It is convenient to use this gate library for didactic reasons, however, the circuit constructed in \lem{cnot} can be rewritten using the same number of entangling gates and the same depth, but relying on the $\cnotgate$ gates only.

\begin{proposition}\label{prop:cnot}
The circuit constructed in \lem{cnot} can be implemented in the same depth and with the same entangling gate count as the original, but using only the $\cnotgate$ gates.
\end{proposition}
\begin{proof}
Given the division of the set of all qubits into two non-overlapping sets $A$ and $B$, a two-qubit gate is called internal to a given set if both qubits it operates on belong to this set and straddling iff it operates over two qubits belonging to different sets.  Clearly, all entangling gates in such circuit are either internal to one of the sets or straddling.  

Choose the sets $A$ and $B$ from the proof of \lem{cnot}.  Observe that we apply Hadamard gates to all qubits in the set $A$ in two layers.  Between those two Hadamard gate layers, all internal gates are $\cnotgate$ gates and all straddling gates are $\czgate$ gates.  This means that we can push the left layer of Hadamard gates to the right layer to cancel both, while flipping controls and targets of some $\cnotgate$ gates and turning $\czgate$ gates into $\cnotgate$ gates using the following rules: 

\begin{center}
	\includegraphics[width=0.2\textwidth]{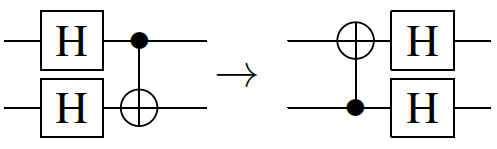} 
	\raisebox{3.5mm}{\,\text{ and}} \\
	\includegraphics[width=0.2\textwidth]{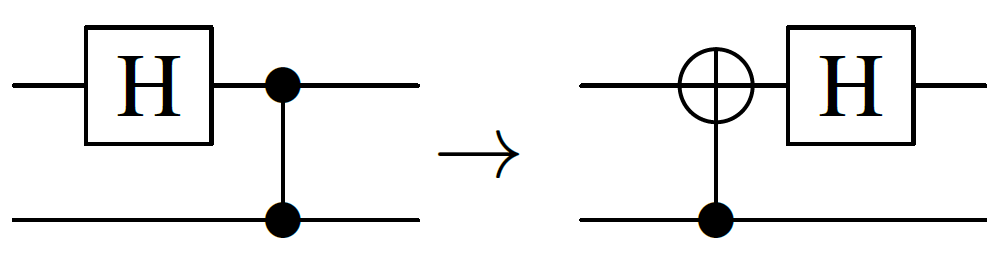}
	\raisebox{3.5mm}{.}
\end{center}

\noindent Observe that this operation, when applied recursively to the matching pairs of layers of Hadamards, eliminates all Hadamard gates and turns all $\czgate$ gates into $\cnotgate$s.  Thus, the transformed circuit has only the $\cnotgate$ gates.
\end{proof}

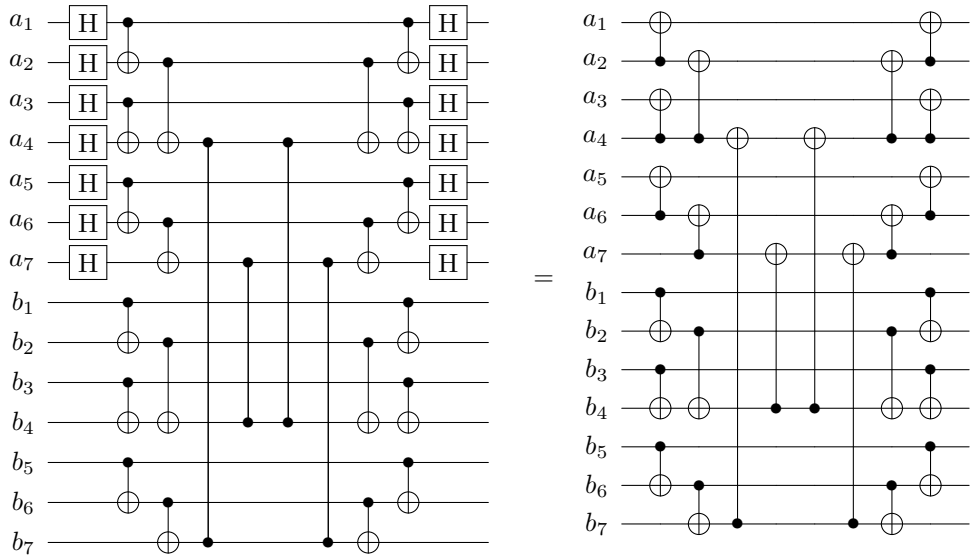
\begin{figure*}[t]
\[ 
\Qcircuit @C=0.1em @R=0.1em @!{
\lstick{a_1} &\gate{\hgate}& \ctrl{1}  & \qw       & \qw      &\qw       & \qw       & \qw       & \qw       & \ctrl{1}  &\gate{\hgate}& \qw \\
\lstick{a_2} &\gate{\hgate}& \targ     & \ctrl{2}  & \qw      &\qw       & \qw       & \qw       & \ctrl{2}  & \targ     &\gate{\hgate}& \qw \\
\lstick{a_3} &\gate{\hgate}& \ctrl{1}  & \qw       & \qw      &\qw       & \qw       & \qw       & \qw       & \ctrl{1}  &\gate{\hgate}& \qw \\
\lstick{a_4} &\gate{\hgate}& \targ     & \targ     & \ctrl{10}&\qw       & \ctrl{7}  & \qw       & \targ     & \targ     &\gate{\hgate}& \qw \\
\lstick{a_5} &\gate{\hgate}& \ctrl{1}  & \qw       & \qw      &\qw       & \qw       & \qw       & \qw       & \ctrl{1}  &\gate{\hgate}& \qw \\
\lstick{a_6} &\gate{\hgate}& \targ     & \ctrl{1}  & \qw      &\qw       & \qw       & \qw       & \ctrl{1}  & \targ     &\gate{\hgate}& \qw \\
\lstick{a_7} &\gate{\hgate}& \qw       & \targ     & \qw      &\ctrl{4}  & \qw       & \ctrl{7}  & \targ     & \qw       &\gate{\hgate}& \qw \\
\lstick{b_1} &\qw   & \ctrl{1}  & \qw       & \qw      &\qw       & \qw       & \qw       & \qw       & \ctrl{1}  &\qw   & \qw \\
\lstick{b_2} &\qw   & \targ     & \ctrl{2}  & \qw      &\qw       & \qw       & \qw       & \ctrl{2}  & \targ     &\qw   & \qw \\
\lstick{b_3} &\qw   & \ctrl{1}  & \qw       & \qw      &\qw       & \qw       & \qw       & \qw       & \ctrl{1}  &\qw   & \qw \\
\lstick{b_4} &\qw   & \targ     & \targ     & \qw      &\ctrl{-4} & \ctrl{-7} & \qw       & \targ     & \targ     &\qw   & \qw \\
\lstick{b_5} &\qw   & \ctrl{1}  & \qw       & \qw      &\qw       & \qw       & \qw       & \qw       & \ctrl{1}  &\qw   & \qw \\
\lstick{b_6} &\qw   & \targ     & \ctrl{1}  & \qw      &\qw       & \qw       & \qw       & \ctrl{1}  & \targ     &\qw   & \qw \\
\lstick{b_7} &\qw   & \qw       & \targ     &\ctrl{-10}&\qw       & \qw       & \ctrl{-7} & \targ     & \qw       &\qw   & \qw
}
\hspace{6mm}\raisebox{-34mm}{=}\hspace{9mm}
\Qcircuit @C=0.64em @R=0.64em @!{
\lstick{a_1} & \targ     & \qw       & \qw      &\qw       & \qw       & \qw       & \qw       & \targ     & \qw \\
\lstick{a_2} & \ctrl{-1} & \targ     & \qw      &\qw       & \qw       & \qw       & \targ     & \ctrl{-1} & \qw \\
\lstick{a_3} & \targ     & \qw       & \qw      &\qw       & \qw       & \qw       & \qw       & \targ     & \qw \\
\lstick{a_4} & \ctrl{-1} & \ctrl{-2} & \targ    &\qw       & \targ     & \qw       & \ctrl{-2} & \ctrl{-1} & \qw \\
\lstick{a_5} & \targ     & \qw       & \qw      &\qw       & \qw       & \qw       & \qw       & \targ     & \qw \\
\lstick{a_6} & \ctrl{-1} & \targ     & \qw      &\qw       & \qw       & \qw       & \targ     & \ctrl{-1} & \qw \\
\lstick{a_7} & \qw       & \ctrl{-1} & \qw      &\targ     & \qw       & \targ     & \ctrl{-1} & \qw       & \qw \\
\lstick{b_1} & \ctrl{1}  & \qw       & \qw      &\qw       & \qw       & \qw       & \qw       & \ctrl{1}  & \qw \\
\lstick{b_2} & \targ     & \ctrl{2}  & \qw      &\qw       & \qw       & \qw       & \ctrl{2}  & \targ     & \qw \\
\lstick{b_3} & \ctrl{1}  & \qw       & \qw      &\qw       & \qw       & \qw       & \qw       & \ctrl{1}  & \qw \\
\lstick{b_4} & \targ     & \targ     & \qw      &\ctrl{-4} & \ctrl{-7} & \qw       & \targ     & \targ     & \qw \\
\lstick{b_5} & \ctrl{1}  & \qw       & \qw      &\qw       & \qw       & \qw       & \qw       & \ctrl{1}  & \qw \\
\lstick{b_6} & \targ     & \ctrl{1}  & \qw      &\qw       & \qw       & \qw       & \ctrl{1}  & \targ     & \qw \\
\lstick{b_7} & \qw       & \targ     &\ctrl{-10}&\qw       & \qw       & \ctrl{-7} & \targ     & \qw       & \qw
}
\]
\caption{Depth-6 implementation of the transformation from Example 1: circuit on the left hand side is obtained by applying \lem{cnot}, and its modification on the right hand side is offered by \prop{cnot}.}
\label{fig:example}
\end{figure*}

We illustrate the constructions in \lem{cnot} and \prop{cnot} with the following Example.
\begin{example}
Consider the $14{\times}14$ linear reversible transformation given by the Boolean matrix 
\[
L\,{:=}\,\left[\begin{smallmatrix}
1 & 0 & 0 & 0 & 0 & 0 & 0 & 1 & 1 & 1 & 1 & 1 & 1 & 1 \\
0 & 1 & 0 & 0 & 0 & 0 & 0 & 1 & 1 & 1 & 1 & 1 & 1 & 1 \\
0 & 0 & 1 & 0 & 0 & 0 & 0 & 1 & 1 & 1 & 1 & 1 & 1 & 1 \\
0 & 0 & 0 & 1 & 0 & 0 & 0 & 1 & 1 & 1 & 1 & 1 & 1 & 1 \\
0 & 0 & 0 & 0 & 1 & 0 & 0 & 1 & 1 & 1 & 1 & 1 & 1 & 1 \\
0 & 0 & 0 & 0 & 0 & 1 & 0 & 1 & 1 & 1 & 1 & 1 & 1 & 1 \\
0 & 0 & 0 & 0 & 0 & 0 & 1 & 1 & 1 & 1 & 1 & 1 & 1 & 1 \\
0 & 0 & 0 & 0 & 0 & 0 & 0 & 1 & 0 & 0 & 0 & 0 & 0 & 0 \\
0 & 0 & 0 & 0 & 0 & 0 & 0 & 0 & 1 & 0 & 0 & 0 & 0 & 0 \\
0 & 0 & 0 & 0 & 0 & 0 & 0 & 0 & 0 & 1 & 0 & 0 & 0 & 0 \\
0 & 0 & 0 & 0 & 0 & 0 & 0 & 0 & 0 & 0 & 1 & 0 & 0 & 0 \\
0 & 0 & 0 & 0 & 0 & 0 & 0 & 0 & 0 & 0 & 0 & 1 & 0 & 0 \\
0 & 0 & 0 & 0 & 0 & 0 & 0 & 0 & 0 & 0 & 0 & 0 & 1 & 0 \\
0 & 0 & 0 & 0 & 0 & 0 & 0 & 0 & 0 & 0 & 0 & 0 & 0 & 1
\end{smallmatrix}\right]
\]
\noindent A na\"ive algorithm focusing on depth optimization may implement this linear transformation in depth 7 by noticing that all off-diagonal ones with matrix indices over non-overlapping sets of qubits can be turned into zeroes by applying the $\cnotgate$ gates with controls in the second half of the set of qubits and target in the first half. A tight schedule exists that squeezes all 49 such $\cnotgate$ gates in depth $49/7=7$. 

A better circuit of depth 6 can be obtained by applying \lem{cnot}.  First, observe that the matrix $L$ is already upper triangular and thus the LU decomposition needs not be developed.  The set $A=\{a_1, a_2, a_3, a_4, a_5, a_6, a_7\}$ contains first $7$ qubits, to which the Hadamard gates are applied, and the set $B=\{b_1, b_2, b_3, b_4, b_5, b_6, b_7\}$ contains the remaining $7$ qubits.  The $7{\times}7$ matrix $R_A$ found in the first quadrant of $L$ gives rise to the $7{\times}7$ all-1 $\czgate$ matrix, and thus the circuit for it can be obtained from \lem{rectangles}.  Recall that this circuit EXORs qubits in the sets $A$ and $B$, applies $\czgate$ gates, and uncomputes the EXORs.  All five stages (opening Hadamards, finding EXOR, applying $\czgate$, uncomputing EXOR, closing Hadamards) are clearly visible in the resulting circuit illustrated in \fig{example} on the left side.  The circuit on the right side of \fig{example} is obtained from the one on the left side by applying \prop{cnot}.
\end{example}

We conclude this subsection with the comparison of the depth of $\cnotgate$ circuits developed in our work to the best known previously in \fig{cnot}.  Similarly to the analogous comparison for $\czgate$ circuits, small values of $n$ reveal a small difference between the exact solution and the upper bound (see \lem{cnot}), that is undetectable by eye over the full range (see \fig{cnot}(b)).  For values of $n$ in the target range, our result improves the best known previously by a factor of almost $4/3$, as expected from the asymptotics.

\subsection{Clifford circuits} \label{sec:clifford}

Recall that a Clifford circuit admits the layered decomposition -X-Z-P-CX-CZ-H-CZ-H-P- \cite{bravyi2021hadamard}.  Adding depths of the implementations of $\czgate$ circuits by \thm{1} (two layers) and $\cnotgate$ circuits by \lem{cnot} (single layer), we obtain the following result.  Note that one of the two -CZ- layers neighbors the -CX- layer, thus allowing to merge the $\cnotgate$ gates used in the largest T-transformation with the -CX- stage; accounting for this results in the reduction of the depth by either $\lceil\log(n/2)\rceil{-}1$ or $\lceil \log \lceil\lceil n/2\rceil/2\rceil \rceil$, depending on the first stage called by the recursion \eq{cz3}.

\begin{lemma}\label{lem:cliff}
For $n\,{\in}\,[43..1{,}345{,}000]$ an $n$-qubit Clifford circuit can be implemented in depth
\[
\lfloor 2n + 2.9487{\cdot}\log^2(n) + 8.4909{\cdot}\log(n) - 44.4798\rfloor.
\]
\end{lemma}

\begin{figure*}[t]
{\centering
\begin{tabular}{cc}
\includegraphics[width=0.48\textwidth]{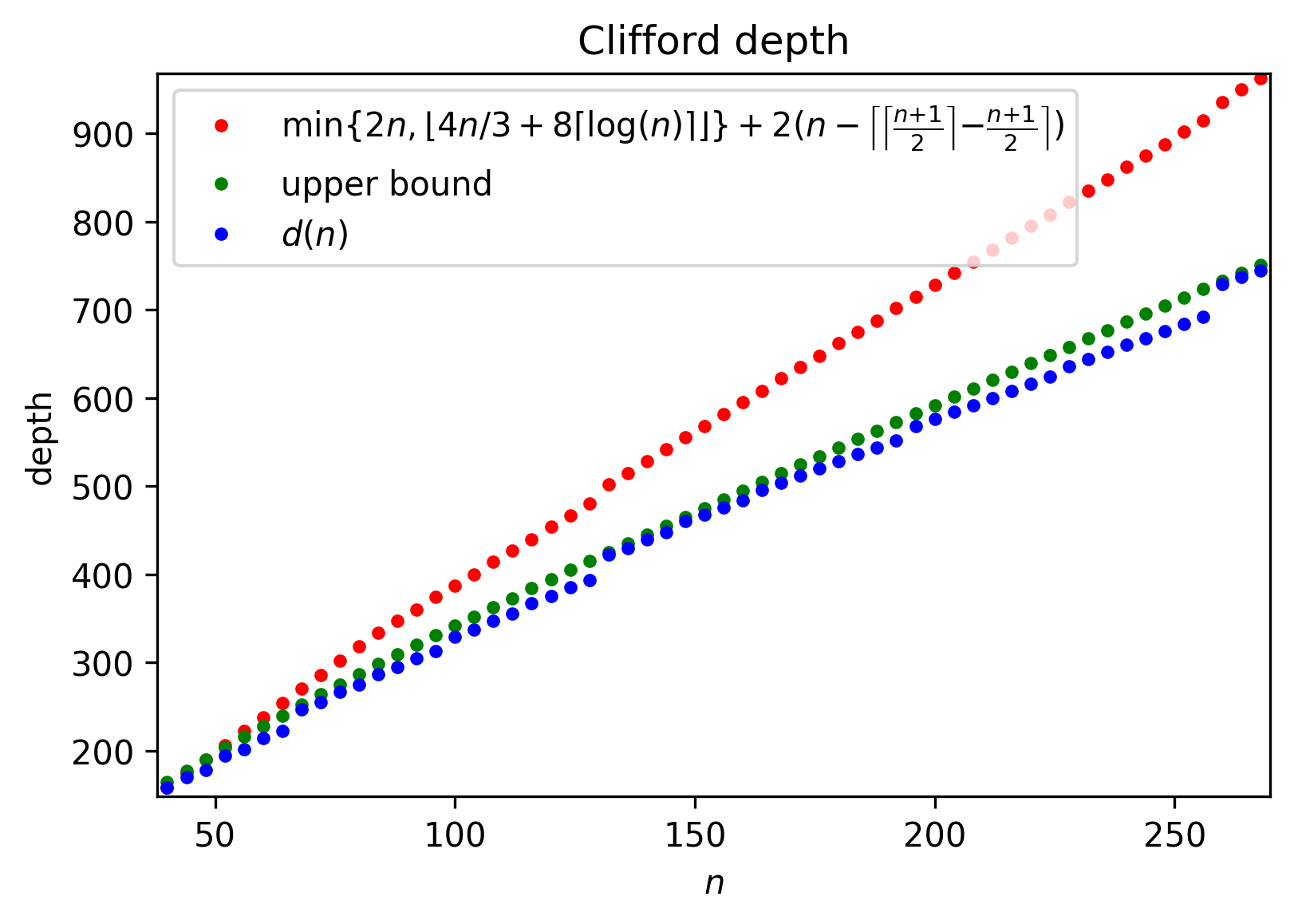} &
\includegraphics[width=0.48\textwidth]{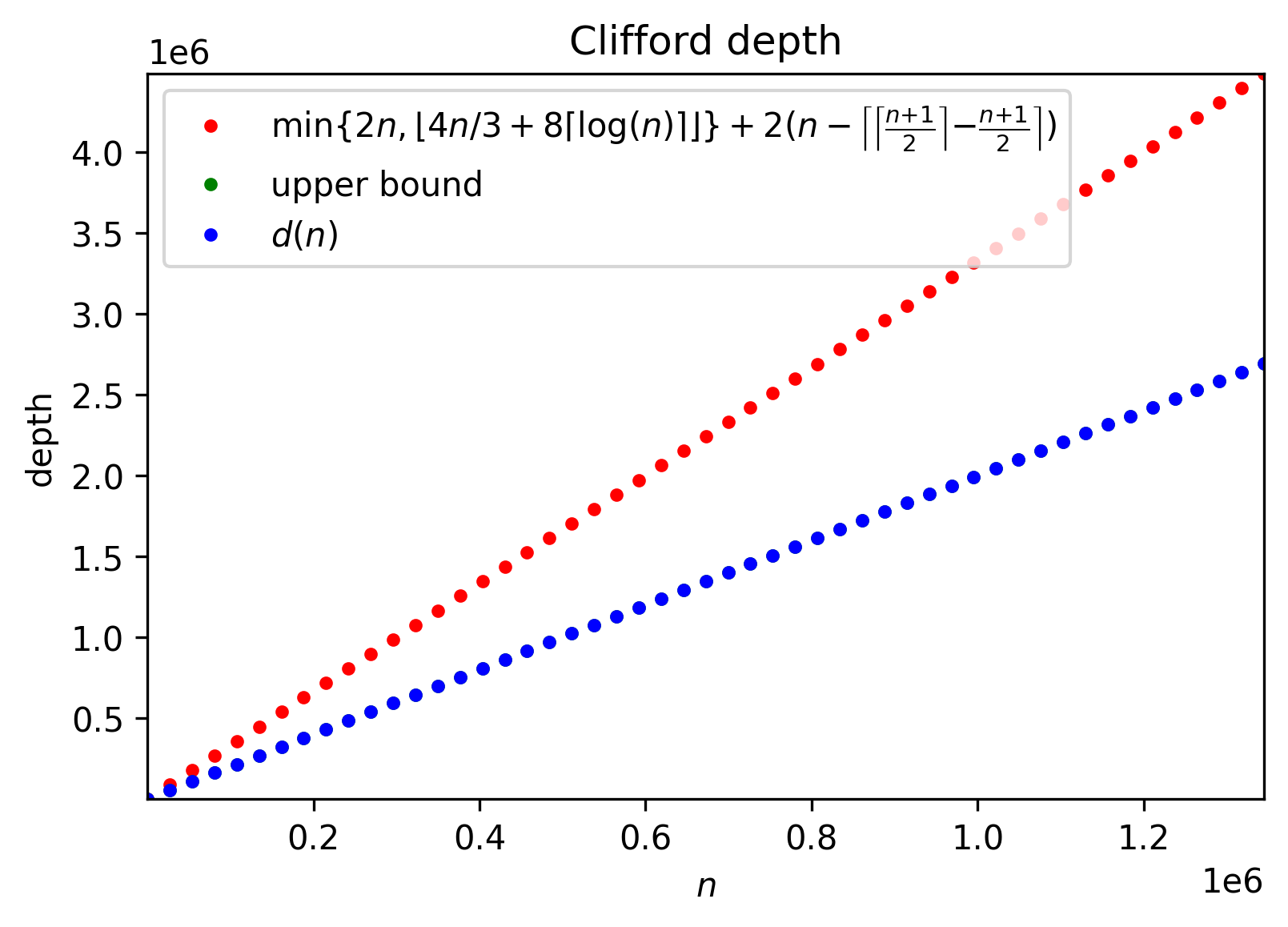} \\
(a) & (b) 
\end{tabular}
}
\caption{Comparison of the best previously known bound on the two-qubit gate depth of Clifford circuits (red dots) to the upper bound established in \lem{cliff} (green dots) and the solution of the respective recursion (blue dots). (a) focuses on a small number of qubits $n\,{<}\,270$ and (b) illustrates the comparison for the full range of values $n$ considered.}
\label{fig:cliff}
\end{figure*}

We illustrated the comparison of the best known depth of Clifford circuits to that offered by our construction, based on the reduced depth of CZ and CNOT circuits (\thm{1} and \lem{cnot}, correspondingly) in \fig{cliff}.

\section{Conclusion}
In this paper, we focused on the study of depth-reduced implementations of $\czgate$ gate quantum circuits spanning a practically relevant number of qubits $n$, $n\,{\leq}\,1{,}345{,}000$.  The improvements in depth were accomplished by implementing $\czgate$ circuits over $\{\czgate,\cnotgate\}$ library rather than relying on the $\czgate$ gates alone.  We extended the methods used to obtain better depth guarantees for $\czgate$ gate circuits to linear reversible and Clifford circuits.  Asymptotic reductions on the $\czgate$, $\cnotgate$, and Clifford circuit depths against state of the art proved in our work are by a factor of $2$, $4/3$, and $5/3$, correspondingly.

\end{document}